\documentclass[reqno,12pt,a4paper]{amsart}

\usepackage{mathrsfs}
\usepackage{amssymb}
\usepackage{amsfonts}
\usepackage{latexsym}
\usepackage{amsthm}

\usepackage{graphicx}
\def\lb{\label}

\newcommand{\er}[1]{\textrm{(\ref{#1})}}

\begin{document}


\renewcommand{\theequation}{\arabic{section}.\arabic{equation}}
\theoremstyle{plain}
\newtheorem{theorem}{\bf Theorem}[section]
\newtheorem{lemma}[theorem]{\bf Lemma}
\newtheorem{corollary}[theorem]{\bf Corollary}
\newtheorem{proposition}[theorem]{\bf Proposition}
\newtheorem{definition}[theorem]{\bf Definition}
\newtheorem{remark}[theorem]{\it Remark}

\def\a{\alpha}  \def\cA{{\mathcal A}}     \def\bA{{\bf A}}  \def\mA{{\mathscr A}}
\def\b{\beta}   \def\cB{{\mathcal B}}     \def\bB{{\bf B}}  \def\mB{{\mathscr B}}
\def\g{\gamma}  \def\cC{{\mathcal C}}     \def\bC{{\bf C}}  \def\mC{{\mathscr C}}
\def\G{\Gamma}  \def\cD{{\mathcal D}}     \def\bD{{\bf D}}  \def\mD{{\mathscr D}}
\def\d{\delta}  \def\cE{{\mathcal E}}     \def\bE{{\bf E}}  \def\mE{{\mathscr E}}
\def\D{\Delta}  \def\cF{{\mathcal F}}     \def\bF{{\bf F}}  \def\mF{{\mathscr F}}
\def\c{\chi}    \def\cG{{\mathcal G}}     \def\bG{{\bf G}}  \def\mG{{\mathscr G}}
\def\z{\zeta}   \def\cH{{\mathcal H}}     \def\bH{{\bf H}}  \def\mH{{\mathscr H}}
\def\e{\eta}    \def\cI{{\mathcal I}}     \def\bI{{\bf I}}  \def\mI{{\mathscr I}}
\def\p{\psi}    \def\cJ{{\mathcal J}}     \def\bJ{{\bf J}}  \def\mJ{{\mathscr J}}
\def\vT{\Theta} \def\cK{{\mathcal K}}     \def\bK{{\bf K}}  \def\mK{{\mathscr K}}
\def\k{\kappa}  \def\cL{{\mathcal L}}     \def\bL{{\bf L}}  \def\mL{{\mathscr L}}
\def\l{\lambda} \def\cM{{\mathcal M}}     \def\bM{{\bf M}}  \def\mM{{\mathscr M}}
\def\L{\Lambda} \def\cN{{\mathcal N}}     \def\bN{{\bf N}}  \def\mN{{\mathscr N}}
\def\m{\mu}     \def\cO{{\mathcal O}}     \def\bO{{\bf O}}  \def\mO{{\mathscr O}}
\def\n{\nu}     \def\cP{{\mathcal P}}     \def\bP{{\bf P}}  \def\mP{{\mathscr P}}
\def\r{\rho}    \def\cQ{{\mathcal Q}}     \def\bQ{{\bf Q}}  \def\mQ{{\mathscr Q}}
\def\s{\sigma}  \def\cR{{\mathcal R}}     \def\bR{{\bf R}}  \def\mR{{\mathscr R}}
\def\S{\Sigma}  \def\cS{{\mathcal S}}     \def\bS{{\bf S}}  \def\mS{{\mathscr S}}
\def\t{\tau}    \def\cT{{\mathcal T}}     \def\bT{{\bf T}}  \def\mT{{\mathscr T}}
\def\f{\phi}    \def\cU{{\mathcal U}}     \def\bU{{\bf U}}  \def\mU{{\mathscr U}}
\def\F{\Phi}    \def\cV{{\mathcal V}}     \def\bV{{\bf V}}  \def\mV{{\mathscr V}}
\def\P{\Psi}    \def\cW{{\mathcal W}}     \def\bW{{\bf W}}  \def\mW{{\mathscr W}}
\def\o{\omega}  \def\cX{{\mathcal X}}     \def\bX{{\bf X}}  \def\mX{{\mathscr X}}
\def\x{\xi}     \def\cY{{\mathcal Y}}     \def\bY{{\bf Y}}  \def\mY{{\mathscr Y}}
\def\X{\Xi}     \def\cZ{{\mathcal Z}}     \def\bZ{{\bf Z}}  \def\mZ{{\mathscr Z}}
\def\O{\Omega}

\def\be{{\bf e}}  \def\bp{{\bf p}} \def\bq{{\bf q}}  \def\br{{\bf r}}
\def\bv{{\bf v}} \def\bu{{\bf u}}
\def\Om{\Omega}
\def\bbD{\pmb \Delta}
\def\mm{\mathrm m}
\def\mn{\mathrm n}

\newcommand{\mc}{\mathscr {c}}

\newcommand{\gA}{\mathfrak{A}}          \newcommand{\ga}{\mathfrak{a}}
\newcommand{\gB}{\mathfrak{B}}          \newcommand{\gb}{\mathfrak{b}}
\newcommand{\gC}{\mathfrak{C}}          \newcommand{\gc}{\mathfrak{c}}
\newcommand{\gD}{\mathfrak{D}}          \newcommand{\gd}{\mathfrak{d}}
\newcommand{\gE}{\mathfrak{E}}
\newcommand{\gF}{\mathfrak{F}}           \newcommand{\gf}{\mathfrak{f}}
\newcommand{\gG}{\mathfrak{G}}           
\newcommand{\gH}{\mathfrak{H}}           \newcommand{\gh}{\mathfrak{h}}
\newcommand{\gI}{\mathfrak{I}}           \newcommand{\gi}{\mathfrak{i}}
\newcommand{\gJ}{\mathfrak{J}}           \newcommand{\gj}{\mathfrak{j}}
\newcommand{\gK}{\mathfrak{K}}            \newcommand{\gk}{\mathfrak{k}}
\newcommand{\gL}{\mathfrak{L}}            \newcommand{\gl}{\mathfrak{l}}
\newcommand{\gM}{\mathfrak{M}}            \newcommand{\gm}{\mathfrak{m}}
\newcommand{\gN}{\mathfrak{N}}            \newcommand{\gn}{\mathfrak{n}}
\newcommand{\gO}{\mathfrak{O}}
\newcommand{\gP}{\mathfrak{P}}             \newcommand{\gp}{\mathfrak{p}}
\newcommand{\gQ}{\mathfrak{Q}}             \newcommand{\gq}{\mathfrak{q}}
\newcommand{\gR}{\mathfrak{R}}             \newcommand{\gr}{\mathfrak{r}}
\newcommand{\gS}{\mathfrak{S}}              \newcommand{\gs}{\mathfrak{s}}
\newcommand{\gT}{\mathfrak{T}}             \newcommand{\gt}{\mathfrak{t}}
\newcommand{\gU}{\mathfrak{U}}             \newcommand{\gu}{\mathfrak{u}}
\newcommand{\gV}{\mathfrak{V}}             \newcommand{\gv}{\mathfrak{v}}
\newcommand{\gW}{\mathfrak{W}}             \newcommand{\gw}{\mathfrak{w}}
\newcommand{\gX}{\mathfrak{X}}               \newcommand{\gx}{\mathfrak{x}}
\newcommand{\gY}{\mathfrak{Y}}              \newcommand{\gy}{\mathfrak{y}}
\newcommand{\gZ}{\mathfrak{Z}}             \newcommand{\gz}{\mathfrak{z}}

\def\ve{\varepsilon} \def\vt{\vartheta} \def\vp{\varphi}  \def\vk{\varkappa}
\def\vr{\varrho} \def\vs{\varsigma}

\def\A{{\mathbb A}} \def\B{{\mathbb B}} \def\C{{\mathbb C}}
\def\dD{{\mathbb D}} \def\E{{\mathbb E}} \def\dF{{\mathbb F}} \def\dG{{\mathbb G}}
\def\H{{\mathbb H}}\def\I{{\mathbb I}} \def\J{{\mathbb J}} \def\K{{\mathbb K}}
\def\dL{{\mathbb L}}\def\M{{\mathbb M}} \def\N{{\mathbb N}} \def\dO{{\mathbb O}}
\def\dP{{\mathbb P}} \def\dQ{{\mathbb Q}} \def\R{{\mathbb R}}\def\dS{{\mathbb S}} \def\T{{\mathbb T}}
\def\U{{\mathbb U}} \def\V{{\mathbb V}}\def\W{{\mathbb W}} \def\X{{\mathbb X}}
\def\Y{{\mathbb Y}} \def\Z{{\mathbb Z}}

\def\dk{{\Bbbk}}


\def\ua{\uparrow}                \def\da{\downarrow}
\def\lra{\leftrightarrow}        \def\Lra{\Leftrightarrow}


\def\lt{\biggl}                  \def\rt{\biggr}
\def\ol{\overline}               \def\wt{\widetilde}
\def\no{\noindent}               \def\ti{\tilde}
\def\ul{\underline}


\let\ge\geqslant                 \let\le\leqslant
\def\la{\langle}                \def\ra{\rangle}
\def\/{\over}                    \def\iy{\infty}
\def\sm{\setminus}               \def\es{\emptyset}
\def\ss{\subset}                 \def\ts{\times}
\def\pa{\partial}                \def\os{\oplus}
\def\om{\ominus}                 \def\ev{\equiv}
\def\iint{\int\!\!\!\int}        \def\iintt{\mathop{\int\!\!\int\!\!\dots\!\!\int}\limits}
\def\el2{\ell^{\,2}}             \def\1{1\!\!1}
\def\wh{\widehat}

\def\sh{\mathop{\mathrm{sh}}\nolimits}
\def\ch{\mathop{\mathrm{ch}}\nolimits}

\def\where{\mathop{\mathrm{where}}\nolimits}
\def\as{\mathop{\mathrm{as}}\nolimits}
\def\Area{\mathop{\mathrm{Area}}\nolimits}
\def\arg{\mathop{\mathrm{arg}}\nolimits}
\def\const{\mathop{\mathrm{const}}\nolimits}
\def\det{\mathop{\mathrm{det}}\nolimits}
\def\diag{\mathop{\mathrm{diag}}\nolimits}
\def\diam{\mathop{\mathrm{diam}}\nolimits}
\def\dim{\mathop{\mathrm{dim}}\nolimits}
\def\dist{\mathop{\mathrm{dist}}\nolimits}
\def\Im{\mathop{\mathrm{Im}}\nolimits}
\def\Iso{\mathop{\mathrm{Iso}}\nolimits}
\def\Ker{\mathop{\mathrm{Ker}}\nolimits}
\def\Lip{\mathop{\mathrm{Lip}}\nolimits}
\def\rank{\mathop{\mathrm{rank}}\limits}
\def\Ran{\mathop{\mathrm{Ran}}\nolimits}
\def\Re{\mathop{\mathrm{Re}}\nolimits}
\def\Res{\mathop{\mathrm{Res}}\nolimits}
\def\res{\mathop{\mathrm{res}}\limits}
\def\sign{\mathop{\mathrm{sign}}\nolimits}
\def\span{\mathop{\mathrm{span}}\nolimits}
\def\supp{\mathop{\mathrm{supp}}\nolimits}
\def\Tr{\mathop{\mathrm{Tr}}\nolimits}
\def\BBox{\hspace{1mm}\vrule height6pt width5.5pt depth0pt \hspace{6pt}}



\def\na{\mathop{\mathrm{\nabla}}\nolimits}

\def\all{\mathop{\mathrm{all}}\nolimits}

\def\arg{\mathop{\mathrm{arg}}\nolimits}
\def\const{\mathop{\mathrm{const}}\nolimits}
\def\det{\mathop{\mathrm{det}}\nolimits}
\def\diag{\mathop{\mathrm{diag}}\nolimits}
\def\diam{\mathop{\mathrm{diam}}\nolimits}



\newcommand\nh[2]{\widehat{#1}\vphantom{#1}^{(#2)}}
\def\dia{\diamond}

\def\Oplus{\bigoplus\nolimits}



\def\qqq{\qquad}
\def\qq{\quad}
\let\ge\geqslant
\let\le\leqslant
\let\geq\geqslant
\let\leq\leqslant
\newcommand{\ca}{\begin{cases}}
\newcommand{\ac}{\end{cases}}
\newcommand{\ma}{\begin{pmatrix}}
\newcommand{\am}{\end{pmatrix}}
\renewcommand{\[}{\begin{equation}}
\renewcommand{\]}{\end{equation}}
\def\eq{\begin{equation}}
\def\qe{\end{equation}}
\def\[{\begin{equation}}
\def\bu{\bullet}
\def\ced{\centerdot}
\def\tes{\textstyle}


\title[{Scattering for anisotropic potentials }]
{Scattering for anisotropic potentials }

\date{\today}
\author[Evgeny Korotyaev]{Evgeny Korotyaev}
\address{Academy for Advance interdisciplinary Studies, Northeast Normal University,
Changchun, China, \  \ korotyaev@gmail.com}


\begin{abstract}
We consider the scattering for the operator  $H=H_o+V$, where the
unperturbed operator  $H_o$ is not assumed to be elliptic  and the
potential $V$ is anisotropic. Under some conditions on $H_o$ and $V$
we show that the wave operators for $H_o, H$ exist and are complete,
$H$ has no singular continuous spectrum and the eigenvalues of $H$
can accumulate only to zero. For stronger conditions on $V$ the
operator $H$ has finite number of eigenvalues only. Moreover, these
results are applied to the invariance principle and for
time-dependent potentials.
\end{abstract}

\maketitle

\section {Introduction and main results}
\setcounter{equation}{0}

\subsection{Introduction}
We consider the scattering for an operator $H=H_o+V$ on $L^2(\R^d)$,
where $ H_o=P(-i\na)$ and $V(x), x\in\R^d$ is an anisotropic
potential satisfying \er{0.1}. We denote the coordinate variable by
$x$, and the conjugate momentum variable by $k$. Here $P(k)$ is a
real function of $k\in \R^d$ satisfying the condition

\no {\bf Condition P.} {\it The function $P(k), k\in\R^d$ is real
and for some $a_j>1$ and $0\le j_-\le j_+\le \n$ has the form}
\[
\lb{P1}
\begin{aligned}
& P(k)=\sum_1^\n p_j(k_j),
\\
& k=(k_j)_{j=1}^\n\in \os_1^\n\R^{d_j}=\R^d,
\end{aligned}
 \qqq
p_j(k_j)=\ca \ \ |k_j|^{a_j}, & j=1,2,..,j_- \\
\  \ |k_j|^{a_j}\sign k_j    &   j_-< j\le  j_+, \ d_j=1\\
  -|k_j|^{a_j}   & j_+<j\le \n\ac .
\]
There are a lot of results about the scattering under the condition
 $V(x)=O(|x|^{-q})$ as $|x|\to \iy$, where $q>1$ see e.g. \cite{RS79}. We call
 such potential isotropic. The completeness of the wave operators  for anisotropic cases was
proved by Deich-Korotyaev-Yafaev \cite{DKY77}, where the main
gradient is imbedding theorems of Korotyaev-Yafaev  \cite{KY77}. Our
goal is to study the scattering for $H_o, H$, singular-continuous
spectrum of $H$ and its eigenvalues. Roughly speaking we use "mixed
approach": the Enss method \cite{E78} for the first variable, Kato's
smooth method \cite{K66} to the second variable plus a priori
estimates. We shortly describe our main results:

\no $\bu$ the wave operators for $H_o, H$ exist and are complete,
 singular-continuous spectrum of $H$ is absent, its eigenvalues can
accumulate only to zero;

\no $\bu$ $H$ has finite number of eigenvalues, 
under stronger conditions on potentials $V$,

\no $\bu$ the invariance principle for anisotropic case,

\no $\bu$  scattering for time-dependent anisotropic potentials.

\subsection {Main results}

Consider the self-adjoint operators $H_o=P(-i\na), H=H_o+V$ acting
on the Hilbert  space $\mH=L^2(\R^d)$, where $P$ satisfies Condition
P and the potential  $V\in \dL_q\cup \cL_\ve$, defined below. Here
$P_{ac}(H_o)=\1$ is the identity operator. For the operators $H_o,
H$ and an open interval $\o\ss \R$ we introduce the wave operators
\[
\lb{1.1} W_\pm(H,H_o,\o)=s-\lim e^{itH}e^{-itH_o}E_o(\o)\qqq as \qqq
t\to \pm \iy,
\]
where $E_o(\cdot)$ is the spectral projector of $H_o$. If $
W_\pm(H,H_o,\o)\mH=E(\o)\mH_{ac}$, where $E(\cdot)$ is the spectral
projector of $H$, then  the wave operators are called  complete. For
the multi-index $\ve=(\ve_j)_1^\n, \ve_j\ge 0$ we define the
function $\vr^{\ve}(x)$ by
$$
\vr^{\ve}(x)=\prod_{j=1}^\n  \vr_j^{\ve_j}(x),\qq \vr_j=\la x_j\ra
:=(1+x_j^2)^{-{1\/2}}, \qqq x=(x_j)_{1}^\n\in \os_1^\n\R^{d_j}=\R^d.
$$
We introduce the space $\cL_\ve$ of anisotropic functions and the
space $\dL_q$ of isotropic functions.

\no {\bf Definition.} {\it i) Let $\cL_\ve $ denote the space of real
functions $f\in L^\iy(\R^d)$ such that
\[ 
\lb{0.1}  
\vr^{-\ve}f\in L^\iy(\R^d),\qqq
 \sup _{|x|=r}|\vr^{-\ve}(x)f(x)|\to 0\qqq {\as} \qq r\to \iy,
\]
where  the multi-index $\ve=(\ve_j)_1^\n $ belongs to $\E_\pm $ or
$\E_o$  given by
$$
\begin{aligned}
\tes \E_\pm=\bigcap_{ j\in \J_\pm} \E_j,\qqq\ \E_o=\{\ve\in
\ol\R_+^\n: \gr>1, \gr_j> {1\/2} \    {\rm for\ some} \ j\in \J \},
\qq\J=\J_\n=\{1,..,\n\},
\end{aligned}
$$
$$
\begin{aligned}
\tes \E_j=\Big\{\ve\in \ol\R_+^\n: \ve_j+\gr-\gr_j>1\Big\}\sm
\Big\{\ve\in \ol\R_+^\n: \ve_j\le {1\/2},\ \ \gr_m\le {1\/2}, \ {\rm
for\ some}\ m\in \J\sm\{j\}\Big\},
\end{aligned}
$$
$$
 j\in \J, \qq \tes \gr=\sum_1^\n\gr_j,\ \ \gr_j={1\/2a_j}\min
\{2\ve_j, d_j\},\qq \J_+=\{1,2,..,j_+\},\qq
\J_-=\{j_--1,j_-,..,\n\}.
$$
ii) If $\ve=q>1$, then we set $\dL_q=\cL_q$.
 Let $\cL_{\ve,q} $ denote the space of real functions
 $f=f_{\ve}+f_{q}$, where $(f_{\ve}, f_{q})\in \cL_{\ve}\ts\dL_{q}$.}

Here and below the multi-index $\ve=(\ve_j)_1^\n \in \ol\R_+^\n$. We
present the main results.

\begin{theorem}
\lb{T1.1}

Consider  $H=H_o+V$ on $L^2(\R^d), d\ge 2$, where the operator $H_o=P(-i\na)$.

\no i) Let the potential $V\in \cL_{\ve,q}, \ve \in
\E_\pm, q>1$ and the interval $\o=\R_\pm$. Then the wave operators $W_\pm
(\o)$ in \er{1.1} exist and are complete, $H$ has no singular
continuous spectrum, its eigenvalues have finite multiplicity and
can accumulate only to zero.

\no ii) If in addition $\ve\in \E_o$, then $H$ has a finite number
of eigenvalues, counted with multiplicity.

\end{theorem}

\no {\bf Remark.} The existence of the wave operators is proved by
the stationary phase method under the condition
$\ve_1+....+\ve_\n>1$, see e.g. \cite{RS79}.  The completeness of the
wave operators needs some additional assumptions \cite{DKY77}, see
Examples in Section 2.
Our conditons on potentials is stronger, in general, than the conditions from \cite{DKY77}, but for specific dimensions $d_j, j\in \N_\n$
they coincide. Here all potentials are bounded, however, the results can be
extended to the case of potentials with local singularities, see Examples in Section 2.

There are many examples which are not cover
 by the previous considerations. We formulate {\it invariance principle} for the following operators
 $$
   T=T_o+V,\qqq T_o=\gf(H_o),\qqq H_o=P(-i\nabla ),
 $$
 where the anisotropic potential $V\in \dL_q$  or  $V\in \cL_\ve$ and the
function $h$ satisfies Condition IP. We denote the range of $P(k),
k\in \R^d$ by $\G$, where $\G=[0,\iy)$ or $\G=\R$.

\no {\bf Condition IP} {\it 1) The real function $\gf\in
L_{loc}^{\iy}(\G)$ and $|\gf(\l)|\to \iy$ as $|\l|\to \iy, \l\in
\G$.

\no  2) Let $\o\ss\G$ be some interval and the mapping  $\gf:\o\to
\O=\gf(\o)$  be a  bijection,  $\gf\in C^m(\o), m>3$ and
$\gf'(\l)>c$ for any $\l\in \o$ and some $c>0$.}

Our second main  result is formulated for the operators
$T_0=\gf(H_o)$ and $T=T_0+V$.

\begin{theorem}
\lb{T1.2} ({\bf Invariance principle.}) Let a function $\gf$ satisfy
Condition  IP and let  a potential $V\in
\cL_{\ve,q}$, where $ \ve \in \E_\pm, q>1$ and an interval  $\o\ss \R_\pm$. Then

\no i) the wave operators $W_\pm (T,T_o,\O)$ exist and are complete,

\no ii) $\s_{sc}(T)\cap \O=\es$ and eigenvalues of $T$,
 belonging to $\O$  have finite multiplicities and can accumulate  only to the ends of the interval $\O$.

\end{theorem}

\subsection {Time-dependent potentials}
Now we consider the scattering for the time-dependent Schr\"odinger
equation with the Hamiltonian $H(t)=H_o+V_t$ on $L^2(\R^d)$ given by
$$
\lb{se1} \tes
 i{\pa u(t)\/\pa t}=H(t)u(t),\qqq u(0)=u_o\in L^2(\R^d),
$$
where $H_o=P(-i\na)$ is the unperturbed operator  as before and the
potential $V_t(x)$ is time-dependent, anisotropic. We assume that
the potential $V_t(x), t\in \R$ satisfies

\no {\bf Condition VT.} {\it The function $V_t(x)\in L^\iy(\R^d\ts
\R) $ is real, the mapping $t\to V_t$ is the continuously
differentiable $L^\iy(\R^d)$--valued function.}

In order to discuss scattering we need to introduce the propagator
for our Hamiltonian.  A {\it propagator} is a two-parameter family
of unitary operators $U(t,s)$, $t, s \in \R$, acting on $\mH$ and
satisfying the following conditions:

{\it
$\bu$ $U(t,s)U(s,r)=U(t,r)$ for all $t,s,r\in \R$,

$\bu$ $U(t,t)=\1$ for all $t\in \R$, where $\1$ is the identity
operator in $\mH$,

$\bu$ $U(t,s)$ is strongly continuous in $t,s\in \R$.}

Our Hamiltonian $H(t)=H_o+V_t, t\in \T$ is a family of
time-dependent operators on $\mH$ and $V_t$ is a bounded operator.
Suppose a propagator $U(t,s)$ leaves $\mD=\mD(H_o)$ invariant and
that
\[
\lb{DD2} \tes   \frac{d}{dt}U(t,s)f=-iH(t)U(t,s)f,\qqq f\in \mD,
\]
 in the strong sense. Then  the family $U(t,s)$ is called the {\it propagator for  $H(t)$}.   In this situation we have
\[
\lb{DD1}  \tes   \frac{d}{ds}U(t,s)f=iU(t,s)H(s)f,\qqq f\in \mD,
\]
where the derivative has the strong sense. Due to well known results
about time-dependent Hamiltonians (see e.g. Theorems X.70, X.71 in
\cite{RS75}) for our Hamiltonian  $H(t)$ under the Condition VT
there exists the propagator $U(t,s)$. The propagator $U_o(t,s)$ for
$H_o$  has the form $U_o(t,s)=e^{-i(t-s)H_o}$.

\begin{theorem}
\lb{Tt2}

Let $H(t)=H_o+V(t)$, where $H_o=P(-i\na)$ and a potential  $V_t$
satisfy Conditions  VT and the following
\[
\lb{tt1}
|V_t(\cdot)|\le g(t)\vr^\ve, \qq  {\rm where} \qq 
\ve \in \E_o,\qqq
\la t\ra^{-\g} g\in L^2(\R),
\]
for some $\g>0$ such that $\g+\gr>{1\/2}$. Then there exist the unitary wave
operators given by
\[
\lb{tt2} \cW_\pm=s-\lim U(0,t)e^{-itH_o}\qqq as \qqq t\to \pm \iy.
\]
\end{theorem}

Now we discuss the scattering for time periodic Hamiltonian
$H(t)=H_o+V_t, t\in \T=\R/\Z$, where $H_o=P(-i\na)$ is the
 operator as before. We assume that the potential
$V_t(x), t\in \T=\R/\Z$   is 1-periodic in time, i.e., $V_{t + 1} =
V_t$ for any $t\in \R$  and    satisfies Condition VT. Here we
introduce the additional periodicity condition for the propagators:

$\bu$  {\it $U(t,s)$ is $1$--periodic, i.e. $U(t+1,s+1)=U(t,s)$ for
all $ t,s\in \R$.}

\no It is well known that for our Hamiltonian  $H(t)$ there exists
the propagator $U(t,s)$, which  is 1-periodic, i.e. it obeys
$U(t+1,s+1)=U(t,s)$ for all $ t,s\in \R$, see \cite{H79},
\cite{Y77}.
 Introduce the monodromy operators $M=U(1,0)$ and $M_o=e^{-iH_o}$
for $H(t), H_o$ respectively. We present our main results about time
periodic scattering.

\begin{theorem}
\lb{Tt1}

Let $H(t)=H_o+V(t)$, where $V_t(\cdot)$ is 1-periodic in time and
obeys Condition   VT  and the following
$$
|V_t(\cdot)|\le F, \qq  {\rm where}  \qq F\in \cL_{\ve,q},
$$
and  $ \ve \in \E_-\cap
\E_+ , q>1$. Then the wave operators
\[
\lb{t2} \cW_\pm=s-\lim U(0,t)e^{-itH_o}\qqq as \qqq t\to \pm \iy,
\]
exist and are complete, i.e., $\cW_\pm\mH=\cH_{ac}(M)$,
$\s_{sc}(M)=\es$ and eigenvalues of $M$
 have finite multiplicity and can accumulate only  to 1.
If   in addition, $F\in \cL_\ve, \ve\in \E_o $, then $M$ has a
finite number of eigenvalues, counted with multiplicity.

\end{theorem}

\subsection{Related works}
 The case of isotropic
potentials and simply characteristic polynomials $P$ was studied by
Agmon and Hormander \cite{AH76}. They show that there  exist the
complete wave operators, singular continuous spectrum is absent and
the point spectrum is investigated. In our anisotropic case the
completeness of the wave operators was proved by Deich, Korotyaev
and Yafaev \cite{DKY77}, where the smooth technique of Kato
\cite{K66}, Lavine \cite{L72} and the imbedding theorems of
Korotyaev and Yafaev from \cite{KY77} were applied.  Later the
scattering for the isotropic potentials and simply characteristic
polynomials $P$ was studied by Muthuramalingam \cite{M85},
\cite{Mu85} and Pascu \cite{Pa91} by the Enss method. Scattering for
Stark operators perturbed by anisotropic potentials was discussed by
Korotyaev and Pushnitski \cite{KP95}. Note that there exists a book
of Perry \cite{P83}, where the Enss method are applied to various
operators.

There are many results  devoted to scattering  mainly for self-adjoint
time-dependent Hamiltonians $H(t)=-\D + V_t(x)$,
 on $\R^d, d\ge 1$. Mainly articles are devoted
 for time-periodic Hamiltonians and to the spectral analysis of the
corresponding monodromy operator. Completeness of the wave operators
for $H_o, H(t)$ was established by Yajima \cite{Y77}. Later on it
was shown that the monodromy operator does not have singular
continuous spectrum, see \cite{H79}, \cite{K80}, \cite{K84} and it
has finite number of eigenvalues \cite{K84}. The case of Schrodinger
operators with time-periodic electric and homogeneous magnetic field
was discussed in \cite{K80}, \cite{K85}, \cite{Y82}, see also recent
papers \cite{AK19,AK16,AK10,Ka19,M00, Y98}. Moreover, scattering for
three body systems was considered in \cite{K85,MS04,Na86}. There are
o lot of results about scattering for discrete Schr\"odinger
operators, see \cite{BS99}, \cite{IK12}. There are only few results
about time periodic Hamiltonians on the graphs
\cite{IK25,K21,K24,K25}.

\section {Preliminaries }
\setcounter{equation}{0}

 Recall
results from \cite{DKY77}, devoted to the more
general case. Below we assume that $\ve\in
\ol\R_+^\n$.

\begin{proposition}
\lb{TKDY}

i) Let $H=H_o+V$, where $H_o=-\D$ and $\o=\R_+$. Let the potential
$V$ belongs to $\dL_q, q>1$ or $\cL_\ve, \ve\in \K_+ $, where $\ve
\in \K_+$ is given by
\[
\lb{1.3}  \K_+=\bigcap_{j\in\J} \K_j,\qq \K_j=\{\ve\in \ol\R_+^\n:
\ve_j+\wt\gr-\wt\gr_j>1\},\qq \wt\gr_j={1\/2}\min\{\ve_j, d_j\}, \ \
\wt\gr=\sum_{i\in\J}\wt\gr_i.
\]
 Then the wave operators
$W_\pm (H,H_o,\o)$ in \er{1.1} exist and are complete.

\no ii) Let $H=H_o+V$, where $H_o=-\D_1+\D_2$ and $\o=\R_+$. Let the
potential $V$ belongs to $\dL_q, q>1$ or $\cL_\ve, \ve\in \K_+ $,
where $\ve \in \K_+$ given by
\[
\lb{1.3x} \tes  \K_+=\{\ve\in\ol\R_+^2:\ve_1+\wt\gr_2>1\}.
\]
 Then the wave operators
$W_\pm (H,H_o,\o)$ in \er{1.1} exist and are complete.

\end{proposition}

Remark that from Proposition \ref{TKDY} and  Theorem \ref{T1.1} we
deduce that $\K_+\subseteq \E_+$, but below we show that $\K_+=\E_+$
for specific cases.

\subsection {Examples}
We illustrate Theorem \ref{T1.1} by few examples.

\no {\bf Example 1.} Let $ H=H_o+V$ in $L^2(\R^d)$, where $H_o=-\D$
and the potential  $V\in \cL_\ve $.

\no A) The general case: $\n\le d$  and $\o=\R_+$. Recall that
$a_j=2$ and $\gr_j={1\/4}\min\{2\ve_j, d_j\}\le {d_j\/4}$ and
$\ve_j\ge 0, j\in \J_\n$ and $\gr=\sum_{}\gr_j$. Here we have set
$\E_+=\cap_1^\n \E_j$, where  $\E_j, j\in \J_\n$ is given by
\begin{multline}
\lb{1.2} \tes \E_j=\Big\{ \ve_j+\gr-\gr_j>1\Big\}\sm \Big\{ \ve_j\le
{1\/2},\ \gr_m\le {1\/2},\ {\rm for \ some }\ m\in \J_\n\sm \{j\}
\Big\},
\end{multline}
so that $\gr_j\le {1\/2}$ if $d_j\le 2$ and $\gr_j>{1\/2}$ if
$d_j\ge 3, \ve_j>1$. Comparing \er{1.2}, \er{1.3} we have that
$\E_+\subseteq  \K_+$. Thus if $\ve\in  \K_+$, then the wave
operators $W_\pm (H, H_o)$ are complete, if in addition $\ve\in
\E_+$, then $\s_{sc}(H)=\es$ and eigenvalues of $H$ can accumulate
only to zero.

\no B) Discuss  the simple case: $\n=d$. Consider $\K_+=\cup \K_j$
for the case $\ve_j=\ve_1<1$ for all $j\in \J_\n$. Thus we have
$\K_j=\{ \ve_j+\wt\gr-\wt\gr_j>1\}=\{ \ve_j=\ve_1>{1\/d}\}$ \ for
all $ j\in \J_\n$. By Proposition \ref{TKDY}, if $\ve\in \K_+$, then
the wave operators $W_\pm (H, H_o)$ exist and are complete. From
\er{1.2} we obtain
$$
\tes \gr_j\le {1\/4},\qq \E_j=\{ \ve_j+\gr-\gr_j>1\}\sm \{ \ve_j\le
{1\/2}\} \ss\{ \ve_j=\ve_1> {1\/2}\},\qqq j\in \J_\n,
$$
which yields $\gr_j={1\/4}, \gr={d\/4},
\E_j=\{4\ve_j>\max\{2,5-d\}\}$. In particular, we have
$\E_+=\{\ve_j>{1\/2},\ \forall \ j\in \J_\n\}$ for $d\ge 3$. By
Theorem \ref{T1.1}, if $\ve\in \E_+$, then $\s_{sc}(H)\cap \o=\es$
and eigenvalues of $H$ can accumulate  only to zero. Note that we
have $\E_+\ss \K_+$.

\no C) For simplification we consider the case $\n=2, \ d=d_1+d_2$.

\no $\bu$ If    $\ve_j\le {d_j\/2}, j=1,2$, then we have
\[
\lb{1.3xx}
\begin{aligned}
 \tes  \E_+=\{ \ve_1+{\ve_2\/2}> 1,\  {\ve_1\/2}+\ve_2>
1\}=\K_+, \qqq \E_o=\{ \ve_1+\ve_2> 2\}.
\end{aligned}
\]

\no $\bu$ If    $d_j>\ve_j> {d_j\/2},  j=1,2$, then $\gr_j=
{d_j\/2}$ and  we have
\[
\lb{1.3xx}
\begin{aligned}
& \tes  \E_1=\{ \ve_1+{d_2\/4}> 1\}\sm \{ \ve_1<{1\/2}, {d_2\/4}\le
{1\/2}\}=\{ \ve_1+{d_2\/4}> 1\},\qq \E_2=\{ \ve_2+{d_1\/4}> 1\},
\\
& \tes \K_+=\{ \ve_1+{\ve_2\/2}> 1,\  {\ve_1\/2}+\ve_2> 1\}\ss
\E_+=\E_1\cap \E_2.
\end{aligned}
\]
 \no $\bu$ If $\ve\in\E_o$, then  by Theorem \ref{T1.1}, the number of
eigenvalues of $H$ is finite and here
$$
\tes \E_o=\{\gr_1+\gr_2>1\}\sm \{ \gr_1\le {1\/2} \ {\rm or}\
\gr_2\le {1\/2}\}=\{\gr_1+\gr_2>1\}.
$$
Then $d=d_1+d_2\ge 5$ and $\max\{\gr_1,\gr_2\}>{1\/2}$. If $\gr_1>
{1\/2}$, then $\ve_1>1, d_1\ge 3$ and we get
$$
\tes \E_o=\{\gr_1+\gr_2>2, \ve_1+{d_2\/2}>2, \ve_2+{d_1\/2}>2\}.
$$

\no {\bf Example 2.} Let $H=H_o+V$, where $H_o=-\D_1+\D_2$ on
$L^2(\R^d)$ and the potential $V\in \cL_\ve $. Consider the case
$\n=2, d=d_1+d_2$ We are interesting in the positive part of the
operator $H$ and thus $\o=\R_+$. By Theorem \ref{TKDY}, if $\ve\in
\K_+=\{\ve_1+\wt\gr_2>1\}$, then  the wave operators $W_\pm (H,
H_o,\R_+)$ exist and are complete.

By Theorem \ref{T1.1}, if $\ve\in \E_+$, then    $\s_{sc}(H)\cap
\o=\es$ and positive eigenvalues of $H$, belonging to $\o$, have
finite  multiplicity and can accumulate  only to zero. Here we have
\[
\lb{1.6} \tes \E_+=\E_1=\{\ve_1+\gr_2>1\}\sm \{ \ve_1\le {1\/2},
\gr_2\le {1\/2}\}=\{\ve_1+\gr_2>1\}.
\]
Note that if $\ve_2\ge d_2$, then we have $\K_+=\E_+$.

\no {\bf Example 3.} We illustrate Theorem \ref{T1.2} by an example.
Let $H_o=-\D, H=-\D+V$ on $L^2(\R^d), d=d_1+d_2$, where the
potential  $V\in \cL_\ve, \ve\in \E_+$ and the set $\E_+$ is defined
in Example 1. Define a function and intervals
$$
h(\l)=(1+\l^4)^r>0, \qqq \l\in \o=\R_+, \ \qqq r>0, \qqq \O=(1,\iy).
$$
Then due to Theorem \ref{T1.2} and Example 1,  for the operators
$T_0=h(H_o), T_1=T_0+V$ we have:

1)  the wave operators $W_\pm (T_1, T_0,\O)$ exist and are complete,

2) $\s_{sc}(T_1)\cap \O=\es$ and eigenvalues of $T_1$, belonging to
$\O$ can accumulate  only at $1$ and $\iy$.

\no {\bf Example 4.} Let $H_o=-\D_1+\D_2, H=H_o+V$ on $L^2(\R^d),
d=d_1+d_2$, where the potential $V\in \cL_\ve, \ve\in \E_+$ and the
set $\E_+$ is defined by \er{1.6}. Define the function and an
interval
$$
h(\l)=\sinh \l, \qqq \l\in \R, \qqq \O=\o=\R_+.
$$
Then due to Theorem \ref{T1.2} and Example 2,  for the operators
$T_0=h(H_o), T_1=T_0+V$ we have:

\no 1)  the wave operators $W_\pm (T_1, T_0,\O)$ exist and are
complete, $\s_{sc}(T_1)\cap \O=\es$,

\no 2)  eigenvalues of $T_1$, belonging to $\O$,  have finite
multiplicity and can accumulate only at 1 and $+\iy$.

\subsection {Preliminary results}
We denote the class of all bounded and compact operators in a
Hilbert space $\mH$ by ${\cB}(\mH)$ and ${\cB}_\iy(\mH)$
respectively. Consider self-adjoint operators $H_o, H=H_o+V$ acting
on the Hilbert space $\mH$, where $V$ is a bounded operator and
$P_{ac}(H_o)=\1$. We formulate conditions on $H_o$ and $V$ and the
interval $\o\ss \R$.

\no {\bf Condition 1.} {\it The operator $V(H_o-i)^{-1} \in
\cB_\iy$.}

\no {\bf 2.} {\it Let an interval $\o\ss\R$. For any $\vp\in
C_0^\iy(\R)$ such that $0\le\vp\le 1$ and $\supp (\vp-1)=\o$   there
exist bounded operators $\z_j, \vp_j, j\in \J_\n$ for some $\n\in
\N$ such that
\[
\lb{2.1} \vp(H_o)=\sum_1^\n \z_j\vp_j,
\]
where all operators $H_o, \z_j, \vp_j, j\in \N_\n$ are commute. For
each $j\in \J_\n$ there exist two bounded operators $Q_j^\pm$ such
that $Q_j^++Q_j^-=\1$ and the following holds true:}
\[
\lb{2.1x} s-\lim  {Q_j^\pm}^* e^{\pm itH_o}=0 \qqq as \qqq t\to \iy.
\]
\no {\bf  3.} {\it For any $(j,f)\in \J_\n\ts \mH$ and every $r\ge
1$ the following estimate holds true
$$
\int_{t>r}\|Ve^{\mp itH_o}\z_j Q_j^\pm f\|dt\le C_r\|f\|,
$$
where the constant $C_r$ does not depend on $f$ and $C_r\to 0$ as $r\to \iy$.

\no {\bf 4.} 
For some $m\in \J_\n$ and any $f\in L^2(\R^d), r\ge 1$ the following
estimate holds true:
\[
\lb{3.6z} \int_{t>r}\|\vr^{\ve}e^{\mp itH_o}Q_m^\pm f\|dt\le
C_r\|f\|,
\]
where the constant $C_r\to 0$ does not depend on $f$ and $C_r\to 0$
as $r\to \iy$. }

From  Condition 1 the standard arguments yield that for any $\vp\in C_0^\iy(\R)$  we have
\[
\lb{co} \vp(H)-\vp(H_o)\in \cB_\iy,
\]
see e.g., \cite{RS75}. From  Conditions 1 and  3 we obtain
\[
\lb{Kpm} K_j^\pm:=\vp(H)(W_\pm(\o) -\1)\z_j Q_j^\pm \in \cB_\iy,
\qqq \forall \  j\in \J_\n.
\]
We describe scattering for $H_o,H$ and the spectrum of $H$.

\begin{theorem}
\lb{T2.1} i)  Let $H=H_o+V$, where $H_o=P(-\na)$ and the potential
$V$ obeys Conditions 1-3 for some interval $\o\ss\R$. Then the wave
operators $W_\pm (H,H_o,\o)$ exist and  are complete,
$\s_{sc}(H)\cap \o=\es$ and eigenvalues of $H$, belonging to $\o$,
have finite multiplicity and can accumulate only at the ends of the
interval $\o$.

\no ii) Let $\ve\in \E_o$, where $\ve_m>{1\/2}$ for some  $m\in
\J_\n$.   Let Conditions 1,4 and \er{2.1x} hold true for $j=m$. Then
operators $H$ has finite number of eigenvalues on any bounded
interval.
\end{theorem}

\no {\bf Proof.} i) Note that the existence of the wave operators is
a simple fact and it is established by the stationary phase method
\cite{RS79}. In the proof we use the known Enss approach from
\cite{E78}, modifed for our case. Let $W_+=W_+(\o)$ for shortness.
We show that $W_+\mH=\mH_{ac}$, the proof for $W_-$ is similar. The
proof is based on the contradiction. Let $f\in \mH_{ac}\om W_+\mH$
and $f\ne 0$. Thus we can  assume that $E(\o)f=f$ for some small
interval $\o\subset \R_+$, the proof for $\o\subset \R_-$ is
similar. Let $0\le \vp\in C_0^\iy(\R)$   be such that
\[
\tes \vp|_{\o}=1,\qqq {\rm and} \qqq \supp \vp\subset \R_+.
\]
Let $f_n=e^{it_nH}f$ for some increasing sequence   $t_n>0, n\in \N$ such that
$t_n\to +\iy$ as $n\to \iy$.  Condition 1 and \er{co} yield
$(\1-\vp(H_o))f_n=(\vp(H)-\vp(H_o))f_n$ and  $\vp(H)-\vp(H_o)\in \B_\iy$ and then
\[
\lb{e1}
\begin{aligned}
&\|f\|^2=\|f_n\|^2=(f_n,\vp(H_o)f_n)+(f_n,(\vp(H)-\vp(H_o))f_n)
=(f_n,\vp(H_o)f_n)+o(1)
\end{aligned}
\]
as $ n\to \iy$. Due to \er{2.1} we have
$(f_n,\vp(H_o)f_n)=(f_n,\sum_{1}^\n \z_j\vp_jf_n)$, where
 \[
\lb{e2}
\begin{aligned}
& (f_n,\z_j\vp_jf_n,)=\sum_{\t=\pm}(f_n,\z_jQ_j^\t \vp_jf_n)
=\sum_{\t=\pm}\Big[(f_n, W_\t\z_jQ_j^\t\vp_j f_n)-(f_n, K_j^\t
f_n)\Big]
\\
&=\sum_{\t=\pm}(f_n, W_\t\z_jQ_j^\t \vp_j f_n)+o(1)
=
(f_n, W_-\z_jQ_j^- \vp_jf_n)+o(1),
\end{aligned}
\]
since $f\in \mH_{ac}(H)\om W_+\mH$ and due to \er{Kpm} we have $K_j^\pm\in\B_\iy$.
Then  Condition 2 implies
$$
(f_n, W_-\z_j\vp_jQ_j^- f_n)=(f, W_- e^{it_nH_o}\z_jQ_j^-
\vp_jf_n)=o(1)
$$
as $n\to \iy$, which gives the contradiction.

In order to show other results we prove that $\dim E(\o)(\mH\om
\mH_{ac})<\iy$. The proof is based on the contradiction. Let $f_n\in
E(\o) (\mH\om \mH_{ac}), n\in \N$ be some orthonormal sequence .
Thus we can  assume that $E(\o)f_n=f_n$. Using above arguments
 and $f_n \bot W_\pm\mH$, we obtain
$$
1=\|f_n\|^2=\sum_{\t=\pm}\sum_{j=1}^\n(f_n, W_\t\z_jQ_j^\t \vp_j
f_n)+o(1)=o(1).\qqq \qqq
$$
ii) Let $\ve\in \E_0$.  We show that the number of eigenvalues
counting multiplicity is finite. It is enough to consider the
interval $\o=(-1,1)$. Let $0\le \vp\in C_0^\iy(\R)$   be such that $
\vp|_{\o}=1$. From Condition 1, \er{co} we have
$\vp(H)V=(\vp(H)-\vp(H_o))V+\vp(H_o)V\in \cB_\iy$. Then
from Condition 4,  we obtain
\[
\lb{e2m} X^\pm:=\pm i\int_0^\iy e^{\pm itH} \vp(H)Ve^{\mp
itH_o}Q_m^\pm dt \in \cB_\iy.
\]
Let $f_n\in E(\o) \mH, n\in \N$ be some orthonormal sequence of
eigenvalues of $H$. If $n\to \iy$, then from \er{e2m} we have the contradiction, since
$$
\begin{aligned}
1=\|f_n\|^2=\sum_{\t=\pm}\Big[(f_n, W_\t Q_m^\t f_n)-(f_n, X^\t
f_n)\Big]=-\sum_{\t=\pm}(f_n, X^\t f_n)=o(1).\qqq \qqq \BBox
\end{aligned}
$$


We  consider the eigenvalues on the interval $\o=(s,\iy)$ for large
$s\gg1$.  We will show that under Condition  $3(\iy)$, the operator
$H$ has not eigenvalues greater than $s$. The proof of the case
$(-\iy, -s)$ is similar.  We introduce the smooth function $\z\in
C^\iy(\R)$ by
\[
\lb{dz}
\z(\l)=\ca 1   & \l\ge 1\\ 0 & \l\le {1\/2}\ac.
\]
 We present conditions for the
case $\o=(s, \iy)$, where $s\ge 1$.

\no {\bf Condition 3($\iy$).} {\it
 For any $(j,f)\in \J_\n\ts \mH$
and $s\ge 1$ large enough the following estimate holds true:
$$
\int_{\R_+}\|Ve^{\mp itH_o}\z(h_j/s)Q_j^\pm f\|dt\le C(s)\|f\|,
$$
where the constant $C(s)$ does not depend on $f$ and $C(s)\to 0$ as
$s\to \iy$.}

\begin{theorem}
\lb{T2.2} Let there exist the wave operators  $W_\pm(H,H_o,\o)$ for
some $\o=(s,\iy)$ and let Conditions 3($\iy$) hold true for some
$s\ge 1$ large enough. Then  the operator $H$ has no eigenvalues
greater than $s$.
\end{theorem}

\no {\bf Proof.} Let $f\in E(H, (b,b+1))\mH$ be an eigenvalues of
$H$ and $\|f\|=1$, where $b\ge 2s$  and $z=b+i$. Then $\|(H-z)f\|\le
2\|f\|$ and we have
$$
\|(\1-\z(H_o/s))f\|\le \|(\1-\z(H_o/s))(H-z)^{-1}\| 2  \|f\| \le
{2\|f\|/s}.
$$
For any $s$ large enough the function  $\z(P/s)$ has a decomposition
\[
\lb{cc1} \z(P/s)=\sum_1^\n \z(p_j/s)\vp_j(s),
\]
for some bounded functions $\vp_j(s, k), (j,k)\in \J_\n\ts \R^d$.
Thus using \er{cc1} and repeating the  arguments from Theorem
\ref{T2.1} we obtain
$$
1=\|f\|=(\z(H_o/s)f,f)+O(1/s)=\sum_1^\n (\z_j(s)\vp_j(s)f,
f)+O(1/s),
$$
$$
(\z_j(s)\vp_j(s)f,f)=\sum_{\t=\pm}\Big[(W_\t\z_j(s)Q_j^\t\vp_j(s)f,f)
+((\1-W_\t)\z_j(s)Q_j^\t\vp_j(s)f,f)  \Big]=o(1),
$$
which yields a contradiction  for $s$ large enough. \BBox

\section {Proof of Theorem \ref{T1.1}}
\setcounter{equation}{0}

We prove Theorem \ref{T1.1}  checking  Conditions 1,2,3, and 3
($\iy$) for the operator $H_o=P(-i\nabla)$ and the potential $V\in \cL_{\ve,q}$, where $ \ve \in \E_\pm, q>1$. We define operators $Q_j^\pm, j\in \J_\n$.

For the case $\n=1$ we define the operator $\cT^\pm$ from \cite{Y82}
such that $\cT^-+\cT^+=\1$ and describe their properties.  Let an
operator $H_o$ act on $L^2(\R^d)$  by
\[
\lb{cc} \tes  H_o=|\D|^{a\/2}, \qqq d\ge 1, \qqq or \qqq
H_o=|\pa|^{a-1}\pa, \qq \pa:=-i{d\/dx},\qqq d=1,
\]
where $a>1$. We denote the Fourier transform of a function $f(x),
x\in \R^d$ by
$$
\wh f(k)=(\F_d f)(k)=(2\pi)^{-{d\/2}}\int_{\R^d}e^{-ikx}f(x)dx,\qq k\in \R^d.
$$
Define the operator $\cT^\pm$ in the spectral
representation of the space $\mH$, when the operator $H_o$ acts on $L^2(\R^d)$
as a multiplication operator by an independent variable.

\no {\bf Definition of $\cT^\pm$.} \no $\bu$ {\it Consider
$H_o=|\D|^{a\/2}, \ d\ge 1$ from \er{cc}. The spectral
representation of $H_o$ is given by $\P=\cU \F_d$, where $\cU:
L^2(\R^d)\to L^2(L^2(\dS^{d-1}), \R_+, dt) $ has the form
$$
(\cU f)(\l,\m)=a^{-{1\/2}}\l^{d-a\/2a}f(\l^{1/a},\m)), \qqq
\qq (\l,\m)\in \R_+\ts\dS^{d-1}.
$$
The operator $\P H_o\P^*$ is the multiplication operator by $\l$.
Let $\gJ$ be the natural imbedding of $L^2(\R_+)$ into $L^2(\R)$.
Thus $\gJ^*$ is the natural restriction map from $L^2(\R)$ into
$L^2(\R_+)$. Let $\c_\pm$ be the characteristic function of the set $\R_\pm$.
We define operators $\cT^\pm$ and $\gF_\pm$ by
\[
\lb{dTpm1}
\cT^\pm=\gF_\pm^* \gF_\pm,\qqq
\gF_\pm= \rt[\c_\pm \F_1^*  \otimes  \1\rt]\gJ^* \cU\F_d: L^2(\R^d)\to
\cH_\pm=L^2(L^2(\dS^{d-1}), \R_\pm, d\s).
\]
Let $\wh\s_\pm$ be the multiplication operator by the variable $\s$ in
$\cH_\pm$.

\no $\bu$ Consider the operator $H_o=|\pa|^{a-1}\pa$
 on $L^2(\R)$ from \er{cc}. Its spectral representation is given by $\P=\cU
\F_1$, where $\cU: L^2(\R)\to L^2(L^2(\dS^{0}), \R_+, dt), \
\dS^{0}=\{-1,1\} $ and
$$
(\cU f)(\l,\n)=a^{-{1\/2}}|\l|^{d-a\/2a}f(\l^{1/a}), \qqq
\l^{1/a}=|\l|^{1/a}\sign \l.
$$
Here $\P H_o\P^*$ is the multiplication operator by $\l$.
Define the operator $\cT^\pm$ in terms $\gF_\pm$ by}
\[
\lb{dTpm2}
\cT^\pm=\gF_\pm^* \gF_\pm,\qqq
\gF_\pm= \c_\pm \F_1^*\cU\F_1: L^2(\R)\to
\cH_\pm=L^2(L^2(\dS^{0}), \R_\pm, d\s).
\]
Let $\wh \s_\pm$ be the multiplication operator by the variable $\s$ in
$\cH_\pm$. Below we need

\begin{lemma}
\lb{T2.3} Let  an operator $H_o=|\D|^{a\/2}$ on $L^2(\R^d), d\ge 1$
or $H_o=|\pa_x|^{a-1}\pa_x$ on $L^2(\R), a>1$. Then for any $\d,t>0,
s\ge 1$ the following estimates hold true
\[
\lb{2.2}\tes \|\la x\ra^{\d} e^{\mp itH_o}\cT^\pm\|\le C \la
t\ra^{\gr},\qqq \gr={1\/a}\min\{\d, {d\/2}\},
\]
\[
\lb{2.3} \|\la x\ra^{\d} e^{\mp itH_o}\z(H_o/s)\cT^\pm\|\le C \la t
\ra^{\d},
\]
\[
\lb{2.4} \|\la x\ra^{\d} e^{\mp itH_o}\z(H_o/s)\cT^\pm\|\le C \la
ts^\t \ra^{\d},\qqq \tes \t={a-1\/a},
\]
where the function $\la x\ra=(1+|x|)^{-1}, \qq x\in \R^d$  and  $\z$ is given by \er{dz}.
\end{lemma}
\no {\bf Proof}. The estimates \er{2.2},  \er{2.3} were proved in
\cite{Y82}. We shall prove \er{2.4} for the case $H_o=|\D|^{a\/2}$
and the sign $"+"$. The proof of other cases is similar. Let
$F(\s,\m)=(\gF_+f)(\s,\m), f\in \mH$. From the definition of $\cT^+=\gF_+^* \gF_+$
we obtain that $g(x):=(e^{-itH_o}\z(H_o/s)\cT^+f)(x)$ has the form
$$
g(x)=C_d\int_0^\iy d\s\int_{\dS^{d-1}}
F(\s,\m)d\m \int_0^\iy e^{i(x,\n)\l^{1/a}-i\l(t+\s)}\z(\l/s)\l^\g
d\l,
$$
where $\g={d-a\/2a}$. Integrating by parts with respect to $\l$ we
obtain
$$
\begin{aligned}
g(x)=C_d\int_0^\iy \!\!\! d\s\int_{\dS^{d-1}}\!\!\! d\m \int_0^\iy
e^{i(x,\m)\l^{1/a}-i\l(t+\s)}{\z(\l/s)\/\l^{1-\g}}
\rt[{\z'(\l/s)\/i\l}-i\z(\l/s)\g+{(x,\m)\l^{1/a}\/a}
\rt]{d\l\/t+\s}.
\end{aligned}
$$
We rewrite this identity in the form
$$
g(x)=-i\rt(e^{-itH_o}\rt[\z'(H_o/s)+\g
{\z(H_o/s)\/H_o/s}\rt]{1\/s}+{1\/s^\t}\sum_1^dx_n e^{-itH_o}(\na_n
/|\na|){\z(H_o/s)\/(aH_o/s)} \rt)\gF_+^*{1\/\wh \s_++t}\gF_+,
$$
where $x=(x_1,...,x_d)\in \R^d, \na=(\na_1,...,\na_d)$.
From here we obtain
$$
\|\la x\ra e^{-itH_o}\z(H_o/s)\cT^+\|\le C\la t s^\t\ra.
$$
Integrating by parts $m=0,1,....$ times we obtain
$$
\|\la x\ra^{m}e^{-itH_o}\z(H_o/s)\cT^+\|\le C\la t s^\t\ra^{m}.
$$
 Then using the interpolation theorem we get \er{2.4}. \BBox

\no \no $\bu$ {\bf Definition of $\cT_j^\pm$ and $Q_j^\pm$.} {\it
For the operator $h_j=p_j(-i\na_j)$ on $L^2(\R^{d_j})$ from \er{P1}
 we define the operator $\cT_j^\pm, j\in \J_+=\{1,...,j_+\}$ by \er{dTpm1},
 \er{dTpm2} in terms of coordinate $x_j\in\R^{d_j}$.}

Lemma \ref{T2.3} gives that for the operator $h_j=p_j(-i\nabla_j)$,
the function $\vr_j=(1+|x_j|^2)^{-{1\/2}}$ and for bounded operators
$\cT_j^\pm, j\in \J_+$ such that $\cT_j^-+\cT_j^+=\1$  there exist
estimates:
\[
\lb{3.2}   \tes  \qqq\qqq\qqq\qqq \|\vr_j^{\ve_j} e^{\mp ith_j}\cT_j^\pm\|\le C \la
t\ra^{\gr_j},\qqq \gr_j={1\/a_j}\min\{\ve_j, d_j/2\},
\]
\[
\lb{3.3} \| \vr_j^{\ve_j} e^{\mp ith_j}\z(h_j/s)\cT_j^\pm\|\le C \la t \ra^{\ve_j},
\]
\[
\lb{3.4} \|\vr_j^{\ve_j} e^{\mp ith_j}\z(h_j/s)\cT_j^\pm\|\le C \la ts^{\t_j} \ra^{\ve_j},
\]
for any $\ve_j,t>0, s\ge1$, where $\t_j={a_j-1\/a_j}$ and some
constant $C$. Define the operators $Q_j^\pm$ by
$$
\ca Q_j^\pm=\cT_j^\pm, & \qq if \qq \ve_j>{1\/2}\\
Q_j^\pm=\cT_m^\pm \ for\ some \ \ m\in \J_+\sm \{j\}&  \qq if \qq
\ve_j<{1\/2},\qq \gr_m>{1\/2}\ac.
$$
We check main Condition 3 for our operator $H_o$ and the potential
$V$.

\begin{lemma}
\lb{T3.2} i) Let $V\in \dL_q, q>1$ or $V\in \cL_\ve, \ve \in \E_j$
for some $j\in \J_+$. Then
for any $r, s\ge 1, \vt>0$ and any $f\in L^2(\R^d)$ the following
estimate holds true:
\[
\lb{3.6} \int_r^\iy\|Ve^{\mp itH_o}\z(h_j/s)Q_j^\pm f\|dt\le
C_1s^{-\vt} r^{-\vt}\|f\|,
\]
where the constant $C_1=C(V)$ depends on  $V$ only.

\no ii) Let $\ve\in \E_o$. Then for some $j\in \J_\n$ and any $f\in
L^2(\R^d)$ the estimate \er{3.6z} holds true.
\end{lemma}
\no {\bf Proof.} i) Discuss the case when $V\in \cL_\ve, \ve \in
\E_j$ and the case $"+"$, the proof of other cases  is similar.
Define a function $\vr_{(j)}^{\ve}=\vr^{\ve} \vr_j^{-\ve_j}$, where
$\vr^\ve=\prod_1^\n \vr_j^{\ve_j}$.

 If $\ve_j>{1\/2}$, then we have $Q_j^+=\cT_j^+$. From
\er{3.3}, \er{3.4}  for $\ve_j={1\/2}+\g+2\d, \g>0, 1>\d>0$ and
$\t_j={a_j-1\/a_j}$ and $F(t):=\|\vr_{(j)}^{\ve}e^{- itH_o}f\|$  we
obtain
$$
\int_r^\iy\|\vr^{\ve}e^{- itH_o}\z(h_j/s)Q_j^+ f\|dt\le
\int_r^\iy\|\la x_j\ra^{\ve_j}e^{- ith_j}\z(h_j/s)\cT_j^+\| F(t) dt
$$
$$
\le C\int_r^\iy t^{\d-\ve_j}(ts^{\t_j})^{-\d}F(t)dt
\le Cr^{-\d} s^{-\d\t_j}\int_r^\iy t^{-{1\/2}-\d}   [
t^{-\g}F(t)]dt.
$$
Then using the Schwartz inequality and \er{5.1} we get \er{3.6}
under the condition $\g+\gr-\gr_j>{1\/2}$. The last condition holds
true if $\ve_j>{1\/2}, \ve_j+\gr-\gr_j>1$, i.e., when $\ve\in \E_j$.
Note that the proof for $V\in \dL_q, q>1$   is similar.

If $\ve_j<{1\/2}$, then we have $Q_j^+=\cT_m^+$ and $\gr_m>{1\/2}$
for some  $m\ne j$.  For $\gr_m={1\/2}+\g+2\d, \g>0, \d>0$ and
$F(t):=\|\vr_{(m)}^{\ve}e^{- ith_j}\z(h_j/s)f\|$ using  \er{3.2}, we
obtain:
$$
\int_r^\iy\|\vr^{\ve}e^{- itH_o}\z(h_j/s)Q_j^+ f\|dt\le
\int_r^\iy\|\vr_m^{\ve_m}e^{- ith_m} \cT_m^+\|\cdot F(t)dt
$$$$
\le C\int_r^\iy  t^{-\gr_m} F(t)dt \le C r^{-\d} \int_r^\iy
t^{-\d-{1\/2}}\big [\la t\ra^{\g}F(t)\big]dt.
$$
From here, using the Schwartz inequality and \er{5.2} we get
\er{3.6} under the condition $\g+\ve_j+\gr-\gr_j-\gr_m>{1\/2}$. The
last condition holds true if $\ve_m>{1\/2}, \ve_j+\gr-\gr_j>1$,
i.e., when $\ve\in \E_j$.

ii) Let $\ve\in \E_0$. Then $\gr> 1$ and some $\gr_m>{1\/2}, m\in
\N$.   Thus we have $\gr_m={1\/2}+\g+2\d$ for some $\g>0, \d>0$. It
follows from \er{3.2} that for any $r\ge 1$ and
$F(t):=\|\vr_{(m)}^{\ve} e^{- itH_o}f\|$:
$$
\int_r^\iy\|\vr^{\ve} e^{- itH_o}\cT_m^+f\|dt \le
\int_r^\iy\|\vr_{m}^{\ve_m} e^{- ith_m}\cT_m^+\|\cdot F(t)dt
$$
$$
\le C\int_r^\iy  \la t\ra^{-\gr_m} F(t)dt \le C\la
r\ra^{-\d}\int_r^\iy  \la t\ra^{-\d-{1\/2}} \rt(\la t\ra^{-\g}
F(t)\rt)dt.
$$
Then using the Schwartz inequality and \er{5.1} we obtain
\er{3.6z} under the condition
$\g+\gr-\gr_m>{1\/2}$ and $\gr_m={1\/2}+\g+2\d$. The last condition holds
true if $\gr_m>{1\/2}, \gr>1$, i.e., when $\ve\in \E_0$. \BBox

\no {\bf Proof of Theorem \ref{T1.1}} i)
 Due to Theorem \ref{T2.1} we need to check  Conditions 1,2,3 for
the operator $H_o=P(-i\nabla)$,   the potential $V\in \cL_\ve,
\ve\in \E_{+}$ and an interval $\o\ss \R_+$. The proof of other
cases  is similar.

\no {\bf 1}. In \cite{DKY77} it was proved that $V(H_o-i)^{-1}\in
\B_\iy$.

 \no {\bf 2}. Let a  function  $\vp\in C_0^\iy(\R)$ be  such
that $\vp|_{\o}=1$ and $\supp \vp\ss \R_+$. From the properties of
$P$ we get the existence of the smooth functions $\vp_j, j\in \J_+$
and the number $s>0$ such that
\[
\lb{3.5} \vp(P(k))=\sum_{j\in\J_+}\z(p_j(k_j)/s)\vp_j(k,s).
\]
Using \er{3.2}, \er{3.5} we obtain that Condition 2 holds true for
the operator $\z_j=\z(h_j/s)$.

\no {\bf 3}. The Condition 3 follows from Lemma \ref{T3.2}.
 We consider the spectrum of $H$ at high energy and
check Conditions 3$(\iy)$. From the properties of $P$ we obtain the
decomposition
$$
\z(P(k)/s)=\sum_{j\in\J_+}\z(p_j(k_j)/s)\vp_j(k,s),\qqq
|\vp_j(k,s)|\le 1,
$$
for some smooth functions $\vp_j(k,s), j\in \J_+$.
For any $r\ge 1$ we have  $V=V_1+V_2$, where
\[
\lb{3.7}
\ca V_1=\c(|x|<r)V(x), \qqq |V_1(x)|\le Cr^{2d}\la x\ra^{-D}\\
|V_2(x)|\le C_2(r)\la x\ra^{-\ve},\qqq C_2(r)\to 0 \qq as \qq |x|\to
\iy \ac,
\]
and the multi-index
$D=2(d_j)_1^d\in \R^d$. Here $\c(A)=1, x\in A$ and $\c(A)=0,
x\notin A$ for some set $A$.  Then for the function $f(t,s)=e^{-
itH_o}\z(h_j/s)\cT_j^+f$ from \er{3.3}, \er{3.4}  we obtain
$$
\int_0^\iy\|Vf(t,s)\|dt\le \int_0^\iy\Big(r^{2d}\|\la
x\ra^{-D}f(t,s)\|+C_2(r)\|\la x\ra^{-\ve}f(t,s)\|\Big)dt\le C(s)
\|f\|,
$$
where $C(s)=Cr^{2d}s^{-\vt_j}+C_2(r), \vt_j=2d_j\t_j>0$. Taking
$s=r^{3d/\d}, \d=\min \vt_j$, we deduce  that $C(s)$ is small.

\no ii)  Lemma \ref{T3.2} gives Condition 4. From Theorems
\ref{T2.1}, \ref{T2.2} we obtain  the proof of ii).
 \BBox

\section {The invariance principle and time periodic potentials}
\setcounter{equation}{0}

\subsection {The invariance principle}

In order to prove the invariance principle we need some estimates.
Let $A, B, V$ be self-adjoint operators in the Hilbert space $\mH$,
where $V$ is bounded. For a function $F\in L^1(\R^2)\cap C(\R^2)$ we
define an operator
\[
\lb{4.1} Z(F)=\int_{\R^2} e^{isB}V e^{i\t A}F(s,\t)dsd\t.
\]
If $F(s,\t)=\d(s)F_1(\t)$, where $F_1\in L^1(\R)\cap C(\R)$, then
the corresponding operator is given  by
\[
\lb{4.2} Z(F)=\sqrt{2\pi} \ V \wh F_1(-A).
\]

\begin{lemma}
\lb{T4.1} Let $\la \t\ra^{-b} F(s,\t)\in L^1(\R^2)\cap C(\R^2)$ or
$\la \t\ra^{-b} F_1(\t)\in L^1(\R)\cap C(\R)$ for some $b\ge 0$. Let
$\e(t)=\|  Ve^{-itA}f \|$ for $(t,f)\in \R\ts\mH$. Then for any
$r\ge 0$ the following estimate holds true:
\[
\lb{4.3} \int_r^\iy\|Z(F)e^{-itA}f\|dt\le C\int_{-\iy}^0\la
|t|+r\ra^{b}\e(t)dt+\la
r\ra^{b}\int_{t>0}\e(t)dt+\int_{2t>r}\e(t)dt.
\]
\end{lemma}
\no {\bf Proof.} Consider $F\in L^1(\R^2)\cap C(\R^2)$, the
proof for another case is similar. It follows from \er{4.1} that
$$
\begin{aligned}
\|Z(F)e^{-itA}f\|=\|  \int_{\R^2}F(s,\t) e^{isB}V e^{i(\t-t) A}dsd\t
f\|
\\
\le \int_{\R^2} |F(s,\t)| \e(t-\t)dsd\t= \int_{\R} g(\t)\e(t-\t)d\t,
\end{aligned}
$$
where $g(\t):=\int_{\R} |F(s,\t)|ds$. From here we obtain
\[
\lb{4.4} G_r=\int_r^\iy\|Z(F)e^{-itA}f\|dt\le \int_r^\iy dt \int_\R
g(\t)\e(t-\t)d\t=\int_\R  \f(u) \int_r^\iy  g(t-u)dt,
\]
where
$$
w(z):=\int_r^\iy  g(t-u)dt=\int_0^\iy  g(t+z)dt=\int_0^\iy dt\int_\R |F(s,t+z)|ds, \qq z:=r-u.
$$
 From the properties of $F$ we obtain
\[
\lb{4.5} w(z)\le C\ca \la z\ra^{b}& , z>0
\\
                      1, & , z<0 \ac.
\]
 From  \er{4.4}, \er{4.5} we get
 $$
G_r\le C \int_{-\iy}^r \e(u)\la r-u \ra^{b}du+ C \int_r^\iy \e(u)du
 $$
$$
\le C \int_{-\iy}^0\e(u)\la r+|u| \ra^{b}du+ C \int_0^r \e(t)\la r-u
\ra^{b}du + C \int_{r}^\iy \e(u)du.
$$
Using this and the estimate
$$
\int_{0}^r \e(u)\la r-u \ra^{b}du\le \int_{0}^{r\/2} \e(u)\la r-u
\ra^{b}du+\int_{r\/2}^{r} \e(t)\la r-u \ra^{b}du
$$
$$
\le C\la r\ra^{b}\int_{0}^{\iy} \e(t)dt+C\int_{r\/2}^{\iy} \e(t)dt
$$
we obtain \er{4.3}. \BBox

We prove the important result about  the operators $T_0=\gf(H_o)$ and $
T=T_0+V$.

\begin{theorem}
\lb{T4.2} ({\bf Invariance principle.}) Let operators $T_o=\gf(H_o),
T=T_0+V$, where  a function $\gf$ satisfy Condition  IP and let
$H_o, V$ satisfy Condition 1-3. Define operators $A, A_o$ by
\[
\lb{4.6} A=\f(T),\qqq A_o=\f(T_o)=H_oE(H_o,\o)
\]
where  $\f$ is the inverse function for $\gf: \o\to \O$. Suppose
that the wave operators $W_\pm (T, T_o,\O)$ and $W_\pm (A, A_o,\o)$
exist and satisfy
$$
W_\pm (T, T_o,\O)=W_\pm (A, A_o,\o)
$$
Then the wave operators $W_\pm (T, T_o,\O)$ are complete,
$\s_{sc}(T)\cap \O=\es$ and eigenvalues of $T$,
 belonging to $\O$ can accumulate  only at the ends of the interval $\O$.
\end{theorem}

\no {\bf Proof.} We check Conditions 1-3 for $T, T_o$.

\no 1) The operator $V(H_o-i)^{-1}$ is compact and by assumption,
$h(\l)\to \iy$ as $|\l|\to \iy, \l\in \G$. This gives
\[
\lb{4.7} V(T_o-i)^{-1}\in \B_\iy.
\]
Using \er{co} and \er{4.7} we obtain
\[
\lb{4.8} \e(A)-\e(A_o)= \e(\f(T))-\e(\f(T_o))\in \B_\iy,
\]
for any $\e\in C_0^\iy(\R)$, which yields Condition 1.

2) Recall that $\o_1\ss \o$ is a close interval and smooth function
$\vp$ satisfies: $\vp|_{\o_1}=1$. Using Lemma \ref{2.1} and \er{4.6}
we obtain
\[
\lb{4.9} \vp(A_o)=\vp(H_o)=\sum_1^\n \z_j\vp_j.
\]
The operator $H_o$ satisfies Condition 2. Then from \er{4.9} we
deduce that the operator $A_o$ satisfies Condition 2.

3) We check the main Condition 3 for $A, A_o, \o$. Introduce the
sufficiently smooth functions
$$
w(\l,\m)=u(\l)u(\m){v(\m)-v(\l)\/\m-\l}, \qq u(\l)=\vp(\f(\l)),\qqq
v(\l)=\e(\f(\l)),\qq  \l,\m\in \o,
$$
and an operator
$$
X=u(T)\big(v(T)-v(T_o)\big)u(T_o).
$$
By the theory of double operator-valued integrals \cite{BS73}, we
obtain
\[
\lb{4.10} X=\iint w(\l,\m)dE_\m V E_\l^o,
\]
where $E_\m, E_\m^o$ are the spectral projector for the operators
$T, T_o $.  Introduce smooth functions
$$
g(\l,\m)=\vp(y_2)\vp(y_1){\e(y_2)-\e(y_1)\/\gf(y_2)-\gf(y_1)}, \qqq
y=(y_1,y_2)\in \O^2.
$$
We have the identity
$$
w(\m,\l)=g(y),\qqq y_1=\f(\l),\qqq y_2=\f(\l),\qqq \l,\m\in \o, \qqq
y\in \O^2.
$$
From the assumption $\gf\in C^{3+\d}(\o), \d>0$ we obtain
\[
\lb{4.11} g\in W_\vt^2(\R^2),\qqq \vt=(\vt_1,\vt_2)\in \R^2,\qqq
\vt_1, \vt_2\ge 0,\qq \vt_1+\vt_2=2+\d.
\]
Here $W_\vt^2(\R^2)$ is the Sobolev space of functions $f(y)$ such
that $\la \t\ra ^{-\vt}\wh f(\t)\in L^2(\R^2), \t\in \R^2$. In
\er{4.11} we take $g$ such that
\[
\lb{4.12} \la \t_1\ra ^{-b}\wh g(\t)\in L^2(\R^2),\qqq
b=\vt_2-\vt_1>1.
\]
Substituting the Fourier integral $g(y)={1\/2\pi}\iint e^{iy\t}\wh
g(\t)d\t$ into the integral \er{4.10} we obtain
\[
\lb{4.13} X=Z(\wh q)={1\/2\pi}\iint \wh g(\t)e^{iA_1\t_2}
Ve^{i\t_1A_0}d\t.
\]
For $f\in \mH, \|f\|=1$, we introduce the function
$$
\gf_j(t)=\|Ve^{-itH_o}\vp(H_o)\z_j(s)Q_j^+ f \|,\qqq j\in \J, \qq
t\in\R.
$$
Using \er{4.12}, \er{4.6}, \er{4.3} we obtain
$$
\int_r^\iy\|V e^{- itA_0}\vp(A_0)\z_j(s)Q_j^+ f\|dt \le \int_r^\iy
\gf_j(t)dt
$$
$$
\le C\int_\R\la |t|+r\ra^{b}\gf_j(t)dt+C \la r\ra^{b}\int_0^\iy
\gf_j(t)dt+C\int_r^\iy \gf_j(t)dt.
$$
Thus it is enough to estimate  the first integral from the RHS:
$$
\int_\R\la |t|+r\ra^{b}\gf_j(t)dt\le C\int_\R\la |t|+r\ra^{b}dt\le
C\la r\ra^{b-1}.\qqq  \qqq \qqq \BBox
$$
 Let $A$ be a self-adjoint operator and let $E(\l)$ be
the corresponding spectral projector. Denote by $\mM(A)$ the set of
all $f\in \mH$ such that $d(E(\l)f,f)=|\gp(\l)|^2d\l$ for some
$\gp\in L^\iy(\R)$. Below we need the well known invariance
principle from \cite{CG76} and \cite{RS79} (Theorem XI.23).

\begin{lemma}
\lb{T4.3} Let $T_o, T_1$ be  self-adjoint operators acting in the
Hilbert space $\mH$. Let a function $\gf\in C^2(\O)$ for some
bounded interval $\O\ss \R$  satisfy:

\no 1) $\gf '\ge \a$ on $\O$ for some constant $\a>0$.

\no 2) Let $\cI\ss \O$ be any close interval and let $\mD$ be a
dense set in $E(T_o, \cI)P_{ac}(T_o)\mH, \mD\ss \mM(T_o)$. For any
$u\in \mD$ the function $w(t)=e^{itT}e^{-itT_o}u, t\in \R$ is
strongly-differentiated and
\[
\lb{4.14} \|w'(\cdot)\|\in L^2(\R),\qqq
 \la t\ra^{-\d}\|w'(\cdot)\|\in L^1(\R),\qq {\rm for\ some}\ \d>0.
\]
Then there exist  wave operators
$$
W_\pm(\f(T), \f(T_o))u=s\lim e^{it\f(T)}e^{-it\f(T_o)}u\qqq \as \qq
t\to \pm\iy,
$$
and $W_\pm(\f(T), \f(T_o))u=W_\pm(T,T_o)u$.
\end{lemma}

\no {\bf Proof of Theorem \ref{T1.2}.} Consider the operator
$H_o=P(-\nabla)$ acting on $\mH=L^2(\R^d), d>1$ and a potential
$V\in \cL_\ve, \ve\in \E_\pm$, the proof of other cases is similar. Recall that a real function $\gf$ satisfies Condition IP
and the interval $\O=\gf(\o)$ where the interval $\o\ss P(\R^d)$. We
apply Lemma \ref{4.3} to the following case: the operator
$T_o=\gf(H_o), T=T_o+V$ and $\f$ is the inverse function of
$\gf(\l), \l\in \o$. Using the stationary phase method we prove
\er{4.14} for the pair $T_o, T$.  From here and from Lemma
\ref{T4.3} we obtain the existence and the identity $W_\pm(\f(T),
\f(T_o))=W_\pm(T,T_o)$. Then from Theorem \ref{T4.2} we get the
proof of Theorem \ref{T1.2}. \BBox

\subsection {Time depending potentials}

We discuss time-decaying potentials.

\no {\bf Proof of Theorem \ref{Tt2}.} Let $\ve\in \E_o$. Note that
the existence of the wave operators is a simple fact due to the
stationary phase method \cite{RS79}. From \er{tt1}, \er{5.1} we
obtain
$$
\begin{aligned}
\|U(0,t)e^{-itH_o}f-f\|^2\le \Big(\int_0^t
\|V_se^{-isH_o}f\|ds\Big)^2\le \Big(\int_0^t g(s)\|\vr^\ve
e^{-isH_o}f\|ds\Big)^2
\\
\le \int_0^\iy  \la s\ra^{-2\g} g^2(s)ds \int_0^\iy \la s\ra^{2\g}
\|\vr^\ve e^{- isH_o}f\|^2ds\le C_1 G(0)\|f\|^2,
\end{aligned}
$$
where $C_1$ does not depend on $f$ and $G(t)=\int_t^\iy  \la s\ra^{-2\g} g^2(s)ds$.
Similar arguments imply for a large time
$$
\begin{aligned}
\|\cW_+f-U(0,t)e^{-it\D}f\|^2\le
\Big(\int_t^\iy   g(s)\|\vr^\ve e^{-isH_o}f\|ds\Big)^2
\\
\le  G(t)
\int_0^\iy   \la s\ra^{2\g}  \|\vr^\ve e^{- isH_o}f\|^2ds\le C_1G(t)\|f\|^2,
\end{aligned}
$$
where $G(t)=o(1)$ as $t\to\iy$. This gives the norm convergence as
$t\to\iy$ and then the wave operator  $\cW_+$ is unitary. The proof
for $\cW_-$ is similar. \BBox

We discuss   time-periodic potentials.

\no {\bf Proof of Theorem \ref{Tt1}.} Consider  $V\in \cL_\ve,
\ve\in \E_+\cap \E_-$, the proof of other cases is similar.
Note that the existence of the wave operators is a simple fact and
it is established by the stationary phase method \cite{RS79}. We
have $\cW_\pm\ss \mH_{c}(M)$, where $M=U(1,0)$. We show that
$\cW_+=\mH_{c}(M)$,  the proof for $"-"$ is similar. The proof is
based on the contradiction. Thus we can assume that $E(\gu,M)f=0$
for some small interval $\gu\ss \dS: =\{|\l|=1\}$, such that $1\in
\gu$, including its small neighborhood and let $ \gu\Subset\gu_1\ss
\dS $ where  for some $\gu_1$. Let $0\le \vt\in C^\iy(\dS)$   be
such that $\vt|_{\gu}=1$ and $ \vt|_{\dS\sm\gu_1}=0$. Let $\e\in
C_o^\iy(\R)$, where $\e=1$ including  small neighborhood of the
point $0$. Define functions $\e_1$ and $\e_2$ by
\[
\lb{te}
\begin{aligned}
\e_1=\vt(e^{i\l}) \e(\l)\in C_o^\iy(\R),\qqq  \e_2=1-\e_1, \qqq
\l\in \R.
\end{aligned}
\]
We  recall the standard fact:
 for any $(\vt, \vp)\in C(\dS)\ts C_0^\iy(\R)$,  we obtain
\[
\lb{tc1} \vp(H_o)(\vt(M)-\vt(M_o))\in \B_\iy,
\]
 see e.g. \cite{Y82}. Recall the well known fact:
 for any finite number of compact operators $\cG_j\in \B_\iy, j\in \J_\n$
   there exist sequences of integers $n_p$ such that
$n_p\to \pm \iy $ as $p\to \pm \iy$:
\[
\lb{mn} \sum_{j=1}^\n\|\cG_j M^{n_p}f\| =o(1)\qqq \as \qq p\to  \pm
\iy,
\]
see \cite{E78}, \cite{Y82}. Then for some subsequence $n:=n_p\to
\iy$ and for $f_n=M^nf$ we have
\[
\lb{te1}
\begin{aligned}
& \|f\|^2=\|f_n\|^2=(f_n,\e_1(H_o)f_n)+(f_n,\e_2(H_o)f_n),
\\
& \e_1(H_o)f_n= \e(H_o)\vt(H_o)f_n=\e(H_o)(\vt(M_o)-\vt(M))f_n=o(1)
,
\end{aligned}
\]
 since $E(\gu,M)f=\vt(M)f=0$ and due to \er{tc1} we obtain $ \e(H_o)(\vt(M)-\vt(M_o))\in \B_\iy $,
for some sequences of integers $n_p$ such that $n_p\to \pm \iy $ as
$p\to  \iy$.  We have the decomposition
 $$
 \e_2=\f+\f_-,\qq \f =\c_+\e_2\in C^\iy(\R), \ \ {\rm and} \ \f_-=\c_-\e_2,
 $$
where  $\c_+(t)=1, t>0$ and  $\c_+(t)=0, t<0$ and $\c_++\c_-=1$.
Here we have a problem with $\f$, since $V_t\f$ is not a compact
operators in $L^2(\R^d)$. We need some modification. We consider the
first term $\f(H_o)f_n$, the proof for $\f_-(H_o)f_n$ is similar. We
have  the decomposition $\f=\z_1\f+\gb_1\f$, where $\z_1=\z(h_1/s)$
and $\gb_1=1-\z_1$. Let $A_1^\pm :=(\cW_\pm-\1)\z_1Q_1^\pm$. Due to
\er{3.4}, \er{3.6} we have $\|A_1^\pm\|=o(1)$ as $s\to \iy$. Then we
have
 \[
\lb{cc3}
\begin{aligned}
& (f_n,\f f_n)=(f_n,\z_1\f f_n,)+(f_n,\gb_1\f f_n),
\\
& (f_n,\z_1\f f_n)=\sum_{\t=\pm}[(f_n,\cW_\t\z_1 Q_1^\t \f
f_n)-(f_n, A_1^\t f_n)]
\\
&=\sum_{\t=\pm}(f_n,\cW_\n\f_jQ_j^\n f_n)+o(1)
=(f_n, \cW_-\f_jQ_j^- f_n)+o(1)\qqq \as \qq n=n_p\to \pm\iy,
\end{aligned}
\]
since $f\in \mH_{ac}(M)\om \cW_+\mH$ and from \er{3.2} and
$M^m\cW_\n=\cW_\n M_o^m$ for all $m\in \Z$ we deduce that $(M^nf,
\cW_-\f_jQ_j^- f_n)=o(1)$. Thus we apply the same procedure for
function $\z_2$ and $\gb_2$ need to consider the sequence
$(f_n,\gb_1\f f_n)$. We have
$$
(f_n,\gb_1\f f_n)=(f_n,\z_2\gb_1\f f_n)+(f_n,\gb_2\gb_1\f f_n)=(f_n,\gb_2\gb_1\f f_n)+o(1).
$$
Repeating procedure we obtain $ (f_n,\f f_n)=(f_n,\gb \f f_n)+o(1)$,
where $\gb=\prod_1^\n \gb_j$. Then the operator  $\gb \f\in
\B_\iy(L^2(\R^d)$, since the functions $\f, \gb$ are compactly
supported and above arguments yield $(f_n,\gb\f f_n)=o(1)$. \BBox

\section {Appendix}
\setcounter{equation}{0}

\begin{lemma}
\lb{T5.1}  Let $\g>0$ and let $\p\in L^2(\R^d)$ and $\gr=\sum_1^\n
\gr_j, \ \gr_j={1\/a_j}\min\{\ve_j, d_j/2\}, \ve_j\ge 0$.

\no  i) If $\g+\gr>{1\/2}$, then
\[
\lb{5.1} \int_\R \la t\ra^{2\g}\|\vr^{\ve}e^{-it H_o}\p\|^2dt
\le C\|\p \|^2.
\]
ii) If $\ve_\ell+\gr-\gr_\ell+\g>{1\/2}$ for some $\ell\in \J_\n$ .
Then for all $s\ge1$ and some $\vt>0$
\[
\lb{5.2} \int_\R \la t\ra^{2\g}\|\vr^{\ve}\c(|k_\ell|>s)e^{-it
H_o}\p\|^2dt \le C s^{-\vt \ve_\ell}\|\p \|^2.
\]
\end{lemma}

\no {\bf Proof.} From strict inequality for $\g, \gr$ it is enough
to prove the estimates for $0<2\g<1, 0<2\ve_j<d_j, j\in \J_\n$ and
for $2\gr<1$ in the case i) and $\ve_\ell+\gr-\gr_\ell<{1\/2}$ in
the case ii). Then since we have $\gr_j={\ve_j\/a_j}, j\in \J_\n$
and we define the vector $A=(A_j)_1^\n\in\R_+^\n$ by
$$
A_\ell=\ca {1\/a_\ell},\ & if \qq \gr+\g>{1\/2}\\
         1     & if \qq \ve_\ell+\gr-\gr_\ell+\g>{1\/2}\ac.
$$
Thus the inequality for $\ve$ we rewrite in the form
$$
\tes (\ve,A)>{1\/2}-\g, \qqq 0<2\ve_j<d_j, \qq j\in \J.
$$
From these strick inequality we get that enough to prove  \er{5.1},
\er{5.2} for
$$
\tes \ve=\ve'+\d d',\qq d'=(d_j)_1^\n,\qq \g>\d,\qq \ve'\in \{t\in
\R_+^\n: (t,A)={1\/2}-\g,\  0<2t_j<d_j, \ \forall \ j\in \J_\n\}
$$
for  $\d>0$ small enough. Roughly speaking we prove \er{5.1},
\er{5.2} for  $\ve$ belonging to the layer
$$
1<2(\ve,A)+2\g <1+ 2\d(d',A),\qqq \d>0.
$$
We introduce functions
$$
f(k,d)={e^{-|k|}\/|k|^d},\qqq      F_d(k,b)=
(2\pi)^{-{d\/2}}\int_{\R^d}e^{-ikx}\la x\ra^{-b}dx,\qqq k\in \R^d,
b>0,
$$
For the function $F_d$ there exists the following estimate (see
\cite{BS73})
\[
\lb{5.3} |F_d(k,b)|\le C |k|^{b}f(k,b), \qqq k\in \R^d, \ d> b>0.
\]
Let $\e$ be the multiplication operator by the function $\e(k)\ge
0$.  We have
$$
\begin{aligned}
Y(\e)=\int_\R\la t\ra^{2\g} \|\la x\ra^{\ve} e^{-itH_o}\e \p\|^2dt=
\int_\R\la t\ra^{2\g} (\la x\ra^{2\ve} e^{-itH_o}\e\p,
 e^{-itH_o}\e\p)dt.
\\
=\int_{\R^{2d}}
F_d(k-p,2\ve)\e(k)\e(p)\wh\p(k)\ol{\hat\p}(p)\int_\R\la t\ra^{2\g}
e^{-it(P(k)-P(p))}dkdpdt
\\
=\int_{\R^{2d}}  F_d(k-p,2\ve)\e(k)\e(p)
\wh\p(k)\ol{\hat\p}(p)F_1(P(k)-P(p),2\g) dkdp.
\end{aligned}
$$
Introduce the function
$$
G(k,p)=|y_0|^{2\g}f(y_0,1)\prod_1^\n |y_j|^{2\ve_j}f(y_j,d_j),\qqq
y=k-p,\  y_j=k_n-p_j,\qq j\in \J_\n.
$$
Thus from \er{5.3} we obtain
\[
\lb{5.4} Y(\e)<C\int_{\R^{2d}} G(k,p) \e(k)\e(p)|\p(k)\p(p)|dkdp.
\]
Suppose that
\[
\lb{5.5} \int_{\R^{2d}} G(k,p) \e(k)\e(p)|\p(k)\p(p)|dkdp<C_1<\iy.
\]
Then using the well know estimates of the integral operators we obtain
\[
\lb{5.6} Y(\e)<C_1\|\p\|^2.
\]
From here we get \er{5.1} for $\e=1$ and \er{5.2} for
$\e=\c(|k_n|>s)$ and $C_1=Cs^{-\vt \ve_\ell}$.

We shall prove \er{5.5}. We introduce the number $\gr'=(\ve',
A)=\sum_{1}^\n\gr_j, $ where $\gr_j'=\ve_j'A_j$ and the functions
$$
b_n(\s)={(d_j/2)-\ve_j-(d_j\gr'-\gr_j')\s\/\g-\s},\qqq
\b_n(\s)={\gr_j'a_j-d_j\gr'+(d_j\gr'-\gr_j')\s\/\g-\s},\qqq \s\in\R.
$$
Note that
\[
\lb{5.7} b_j(0)>0,\qqq \b_j(1)>0,
\]
and
\[
\lb{5.8}
b_j+\b_j-d_j=\ca  0 &  if \qq j\ne \ell\\
               0 &  if  \qq j= \ell,\  \gr+\g>{1\/2}\\
               {\ve_\ell'(a_\ell-1)\/\g-\s}&  if  \qq j= \ell,\
               \gr+\g+\ve_\ell-\gr_j>{1\/2}\ac.
\]
We show that for any $\ell\in \J_\n$ there exists $\vt_j\in (0,1)$
such that
\[
\lb{5.9} b_j=b_j(\vt_n)\in (0,d_j), \qqq  \b_j=\b_j(\vt_j)>0.
\]
Let $\s_j$ be the solution of the equation $b_j(\s)=0$. There two
cases:

1) $0<\s_j\le 1$. Then there exists some $\vt_j<\s_j$ such that  the
identity \er{5.9} holds true, since $\b_j(\s_j)\ge d_j$.

2) $\s_j> 1$. Then  $b_j(1)>0$ and due to $\b_j(1)>0$ there exists
some $\vt_j<(0,1)$ such that the identity \er{5.9} holds true.

Let $r_j=|k_j|$ and we introduce functions
$$
G_0(k,p)=\prod_1^\n |y_j|^{-b_j}r_n^{-\b_j}e^{-|y_j|},\qq
 G_n(k,p)=f(y_0,1) |y_j|^{\vt_n}r_j^{a_j-\vt_j} \prod_1^\n\
 |y_n|^{-d_n\vt_n}r_n^{d_n(\vt_n-1)} e^{-|y_n|}.
$$
We have the identity
\[
\lb{5.10} G(k,p)=\e(k)\e(p)G_0^{2(\g-\d)}(k,p)\prod_1^\n
G_j^{2\gr'}(k,p)e^{-\d |y_j|}.
\]
For the function $G_0$ we obtain the estimate:
\[
\lb{5.11} G_0(k,p)\le
|y_\ell|^{-b_\ell}s_\ell^{-\b_\ell}e^{-|y_\ell|}\prod_{j\ne \ell}^\n
|y_j|^{-b_j}r_j^{-\b_j}e^{-|y_j|},\qqq if \qq |k_\ell|\ge s.
\]
We have the simple estimate
\[
\lb{5.12} \tes \sup_{p\in\R^d} \int_{\R^d} G_0^\t(k,p)dk <\iy \qqq
if \qq \gr+\g>{1\/2}, \ \t\in (0,1),
\]
and using \er{5.11} we have for $\gr+\g+\ve_\ell-\gr_\ell>{1\/2}$
\[
\lb{5.13} \sup_{p\in\R^d} \int_{\R^d}\c(|k_\ell|>s) G_0^\t(k,p)dk
<\iy \qqq if \tes \qq \gr+\g>{1\/2}, \ \t\in (0,1).
\]
From Lemma \ref{T5.2} we obtain
\[
\lb{5.14} \sup_n\sup_{p\in\R^d} \int_{\R^d} G_n^\t(k,p)dk <\iy \qqq
if \qq \tes \gr+\g>{1\/2}, \ \t\in (0,1).
\]
From \er{5.10} we get
\[
\lb{5.15}
 \int_{\R^d} G(k,p)dk < \int_{\R^d} \e(k)\e(p)G_0^{2(\g-\d)}(k,p)
 \prod_1^\n G_j^{2\gr_j'}(k,p)dk.
\]
From \er{5.12}-\er{5.14} and from an estimate
$$
2(\g-\d)+2\sum_1^\n \gr_j=2(\g-\d+\gr')=1-2\d<1,
$$
we see that we can apply the Holder inequality to the right hand
side  of \er{5.15} with $\e=1$ in the case \er{5.15} and
$\e=\c(|k_\ell|>s)$ in the case \er{5.2}. \BBox

\begin{lemma}
\lb{T5.2}
Let $k, p\in \R^d, u\in \dS^{d-1}$ and $r=|k|, k=ru $.

\no 1) Let $0<b<d-1$. Then
\[
\lb{5.16} X(r)=\int_{\dS^{d-1}} {e^{-|ru-p|}\/|r\n-p|}du
\le C{e^{-{1\/2}|r-|p||}\/r^b+r^{d-1}},
\]
where $C$ does not depend on $p$.

\no 2)  Let $b+\b\le d-1, b,\b\ge 0, a\ge 1$,  and let the function
$$
G(k,p)= {e^{-|r^a-y|-|k-p|}\/|r^a-y||r-p|^br^{1+\b-\a}}, \qq y\in \R.
$$
 Then
\[
\lb{5.17}         \sup_{p,y\in \R^d} \int_{\R^d}  G^\t(k,p)dk<\iy
\qq \forall \ \t\in (0,1).
\]
\end{lemma}

\no {\bf Proof.}  Let $\vt$ be an angle between vectors $k,p$.  Then
\[
\lb{5.18}
|r-p|\ge r\sin \vt, \qqq 0\le \vt \le \pi,
\]
\[
\lb{5.19}
\s \vt\le \sin \vt \le {\vt\/\s}, \qqq 0\le \vt \le {\pi\/2},
\]
for some $\s>0$. From \er{5.18}, \er{5.19} we obtain
$$
X(r)\le C e^{-|r-|p||/2} \int_{0}^\pi  e^{-{r\/2}\sin \vt }
 {\sin^{d-2} \vt\/|r\sin \vt|^{b}  }  d\vt
\le C r^{-b}e^{-|r-|p||/2} \int_{0}^\pi  e^{-{r\s \vt\/2} } \vt^{d-2-b} d\vt.
$$
This yields
$$
X(r)\le C r^{-b}e^{-|r-|p||/2} ,\qqq 0<r\le 1,
$$
and
$$
X(r)\le C r^{1-d}e^{-|r-|p||/2} \int_{0}^\iy  e^{-t}  t^{d-2-b} dt\le
 C r^{1-d}e^{-|r-|p||/2}.
$$
We shall prove \er{5.17}. We have
$$
Y(p,y)=\int_{\R^d}  G^\t(k,p)dk=\int_0^\iy {e^{-\t
|r^a-y|}\/|r^a-y|^\t r^{\t(1+\b-\a)}} \int_{\dS^{d-1}}
{e^{-\t|ru-p|}\/|ru-p|^\t}du\ r^{d-1}dr.
$$
Using \er{5.16} we obtain
$$
Y(p,y)\le \int_0^\iy    {e^{-\t |r^a-y|}\/|r^a-y|^\t
r^{\t(1+\b-\a)}} {r^{d-1}dr\/r^{\t b+r^{d-1}}}\le C\le \int_0^\iy
{e^{-\t |r^a-y|}\/|r^a-y|^\t r^{\t(1+\b-\a)}}{r^{d-1}dr\/r^{\t
b+r^{d-1}}}.
$$
Using the new variable $t=r^q$ we rewrite the last integral in the
form $ \int_0^\iy {e^{-\t |t-y|}\/t^\b |t-y|^a}dt\le C<\iy$, where
$\b=(1-\t)(1-{1\/a})$, since   $\b+\t={a+\t-1\/a}<1$, and $C$ does
not depend on $p,y$. \BBox

\end{document}